\documentclass[a4paper,UKenglish,cleveref, autoref, thm-restate]{oasics-v2021}

\usepackage{algorithm, algorithmic, amsmath, amssymb, float}
\usepackage[title]{appendix}
\usepackage{graphicx} 
\usepackage{times}


\title{A Graph Matching Based Approach for the Multi-Depot Capacitated Vehicle Routing Problem}
\titlerunning{MDCVRP} 
\author{Jayant Chandwani$^*$, Pranav M R$^*$, Anand Jat, Anshu Ostwal, Diptendu Chatterjee, Anand Narasimhamurthy}{Department of CSIS, BITS Pilani K K Birla Goa Campus, Goa - 403726, India}{f20230356@goa.bits-pilani.ac.in, f20230340@goa.bits-pilani.ac.in, diptenduc@goa.bits-pilani.ac.in, anandn@goa.bits-pilani.ac.in}{}{}
\date{}

\authorrunning{Jayant, Pranav, Anand, Anshu, Diptendu, Anand} 

\Copyright{Jayant, Pranav, Anand, Anshu, Diptendu, Anand} 



\nolinenumbers 
\hideOASIcs 

\begin{document}
	\maketitle

    \def\thefootnote{*}\footnotetext{Jayant Chandwani and Pranav M R are the lead authors and contributed equally to this work.}\def\thefootnote{\arabic{footnote}}

	\begin{abstract}
		The Multi-Depot Capacitated Vehicle Routing Problem (MDCVRP) asks for minimum-cost delivery tours from several capacitated depots to a set of customers. Like most vehicle-routing variants it is NP-hard, so practical solvers must trade solution quality against speed. We revisit this trade-off through the lens of graph matching. Adapting a matching-based construction first developed for the Traveling Tournament Problem, we present two algorithms, Cluster-First and Match-First, that reduce routing to a sequence of minimum-weight matchings. This is more than a heuristic. We prove that for tours of up to two targets the matching formulation solves the MDCVRP exactly in polynomial time for any number of depots, and that both algorithms are constant-factor approximations, with a tight factor of two, in the structured regimes. This matching optimum coincides with the exact combinatorial-auction optimum, so the auction serves as a strong quality baseline. On instances of $1000$ customers and $20$ depots our methods match or slightly beat that baseline in tour length while running two to three orders of magnitude faster, in tens of milliseconds against tens of seconds, a scale at which exact and auction-based solvers become impractical. Because Cluster-First routes each depot independently, the approach also re-routes cheaply when new customers arrive.\\
        Code for reproducibility: \url{https://github.com/Floret-Labs/graph-matching-mdcvrp}
\end{abstract}

	\section{Background and Related Work}
        The Vehicle Routing Problem (VRP), introduced by Dantzig and Ramser~\cite{Dantzig1959}, asks for a set of minimum-cost routes, each starting and ending at a depot, that serve every customer exactly once~\cite{Hrvoje2008}.
        Practical variants add constraints such as limited vehicle capacities (the Capacitated VRP, CVRP), delivery time windows, or several geographically distributed depots (the Multi-Depot VRP). This work concerns the \emph{Multi-Depot Capacitated VRP} (MDCVRP), which combines capacitated vehicles with multiple depots, and which we define formally as follows.
\begin{itemize}
\item {\bf Inputs:}
\begin{itemize}
\item A set of delivery vehicles $R = \{r_1,r_2,\ldots,r_K\}$ with corresponding capacities $u_1,u_2,\ldots,u_K$ and initial depot locations $l_{r_1},l_{r_2},\ldots,l_{r_K}$.
\item A set of customers $C = \{c_1,c_2,\ldots,c_n\}$ with corresponding target locations $l_{t_1},l_{t_2},\ldots,l_{t_n}$.
\item A set of non-negative weights $W = \{w_1,w_2,\ldots,w_n\}$, where $w_i$ represents the total weight of the goods in the order placed by customer $c_i$, for $1 \le i \le n$.
\end{itemize}
\item {\bf Output:}
\begin{itemize}
\item A set of tours that originate from and terminate in depots such that each target is visited only once, all customer orders are processed, and an objective is optimized.
\begin{itemize}
\item Objectives that are commonly considered for minimization are the~\emph{sum of the tour costs} and the~\emph{makespan}. Other minimization objectives could include the total number of vehicles and/or penalties for not servicing customer orders.
\end{itemize}
\end{itemize}
\item The VRP problem is typically modeled using a graph where the vertices represent the depot and target locations and the edge weights represent the travel costs between them, assumed to be known beforehand. For convenience, we assume that all locations lie on a 2D Cartesian grid and that all pairwise travel costs are straight-line distances. This assumption can be relaxed without any changes to the approach, since the  input is a graph where edge weights represent the relevant costs.
\end{itemize}

\noindent Like most of its variants, the MDCVRP is NP-hard~\cite{Lenstra1981,Solomon1988,Dror1990,Archetti2005}, so exact methods are impractical beyond small instances and the literature is dominated by heuristics that trade optimality for speed~\cite{ClarkWright,Christofides1976,Khairy2020,bm04,Pu2022,Yesodha2022,Gu2022,Sadati2021}. A recurring strategy for multi-depot instances is \emph{cluster-first, route-next}: group customers by depot so that each cluster's demand fits one vehicle, then route each cluster as a single-depot problem, for example with a nearest-neighbour rule or a general-purpose solver such as Google OR-Tools~\cite{GoogleORTools}. More recently, learning-based methods~\cite{Xin2021_AAAI} have been brought to the multi-depot setting, including the multi-agent reinforcement-learning model of Arishi et al.~\cite{Arishi2023} and the attention-based DeepMDV~\cite{Nasehi2024_DeepMDV}. These produce high-quality solutions on small instances but scale poorly and are expensive to train, which limits their use on the large, dynamically changing instances that arise in practice.

Our approach instead casts routing as graph matching: rather than building routes sequentially, delivery decisions are reduced to pairing targets on a weighted graph and computing minimum-weight matchings with Edmonds' Blossom algorithm~\cite{Edmonds1965_Blossom}. This adapts to the MDCVRP the matching-based construction proposed for the Traveling Tournament Problem~\cite{Diptendu2021}. As a strong point of comparison we additionally formulate the MDCVRP as a combinatorial auction, in which depots bid on small bundles of targets; auction-based formulations have previously been studied for vehicle routing~\cite{Feng2011,Karels2020}.

\paragraph*{Contributions} We present two matching-based algorithms for the MDCVRP, Cluster-First and Match-First, and show that they are near-optimal yet extremely fast. Concretely, we:
\begin{itemize}
    \item Prove that for tours of up to two targets the matching formulation solves the MDCVRP exactly in polynomial time for any number of depots, and that both algorithms are constant-factor approximations, with a tight factor of two, in the structured regimes;
    \item Show that this matching optimum coincides with the exact combinatorial-auction optimum, so the auction is a genuinely strong quality baseline rather than a weak one; and
    \item Demonstrate empirically that the matching methods lose little tour quality relative to the auction while running up to two to three orders of magnitude faster and scaling to instances with $1000$ customers and $20$ depots.
\end{itemize}
Because Cluster-First routes each depot independently, the approach also re-routes cheaply when new customers arrive, the setting that motivates this work. We implement all algorithms in C++.

 	\section{Preliminary Analysis: Selecting a Baseline}\label{sec:prelim}
Our aim is to show that graph matching produces near-optimal routings far faster than conventional solvers, and a claim like that is only as convincing as the baseline it is measured against. Before presenting our own methods, we run a short study whose sole purpose is to fix that baseline: we compare three standard ways of solving the MDCVRP and carry only the strongest forward as the reference point for the main experiments (\cref{sec:results}). All three follow the familiar cluster-first-route-next pattern (assign each target to a depot, then route within each depot) and differ only in how they route each single-depot subinstance:
\begin{itemize}
    \item {\bf Nearest-neighbor (NN)}: solve each single-depot instance greedily, starting from the full vehicle capacity and repeatedly appending the nearest as-yet-unserved target whose demand still fits.
    \item {\bf Two-stage (OR-Tools)}: hand each single-depot instance to the vehicle-routing solver of Google OR-Tools~\cite{GoogleORTools}.
    \item {\bf Combinatorial Auction (CA)}: our own auction formulation (\cref{sec:combinatorialAuction}), in which depots bid on small bundles of targets and a winner-determination step selects a feasible lowest-cost assignment.
\end{itemize}
We put the three contenders to two tests on the base datasets (\cref{sec:exp-setup}):

\noindent \textbf{Does the Combinatorial Auction yield shorter tours than the standard baseline heuristics?}\label{sec:rq1-baseline}

\noindent Figure~\ref{fig:tourLengths_Cap100_Cap150} (a) and (b) summarise the tour lengths obtained across the 100 base instances at two representative capacities. At both capacities the Combinatorial Auction produces consistently shorter tours than both the nearest-neighbor heuristic and the two-stage OR-Tools solver, indicating that the extra modelling effort of the auction translates into a real gain in solution quality.

\noindent \textbf{Is that advantage robust as the vehicle capacity varies?}\label{sec:rq2-baseline-capacity}

\noindent To confirm that this ranking is not an artefact of one capacity setting, we recompute the tours over a range of vehicle capacities, holding the target and depot locations fixed, and average over the 100 datasets (Figure~\ref{fig:capacity_baseline}). The Combinatorial Auction keeps a clear advantage across the whole range. The lone exception is the regime in which capacities dwarf the target demands: there the nearest-neighbor heuristic can occasionally match or beat the auction, simply because it may pack many targets into

 \begin{figure}[H]
 \centering
 \begin{tabular}{cc}
    (a) Capacity = 100 & (b) Capacity = 150 \\
    \includegraphics[width=0.45\textwidth]{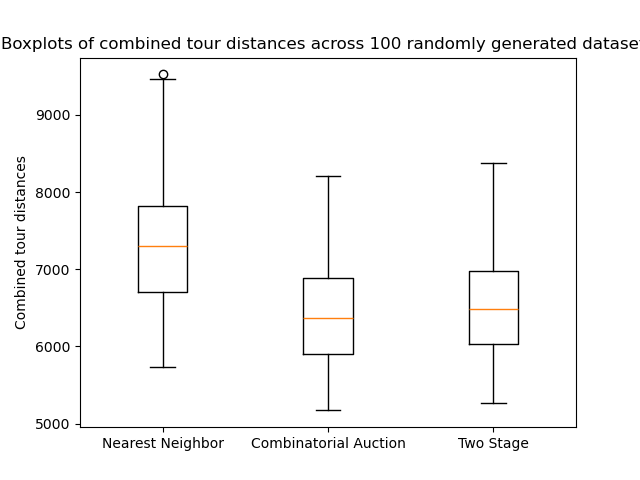} & \includegraphics[width=0.45\textwidth]{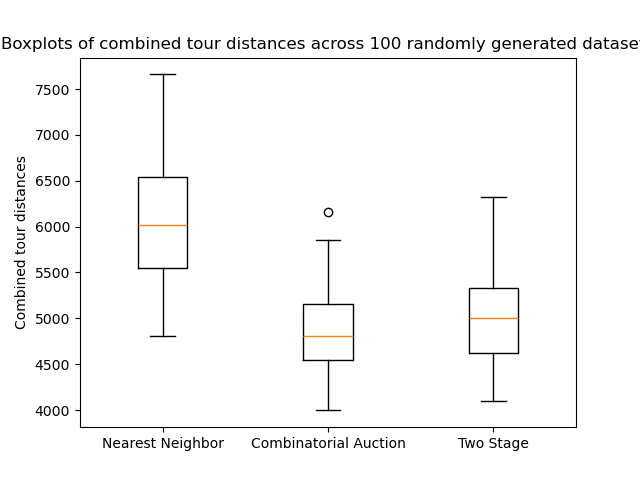}  \\
 \end{tabular}
 \caption{Tour lengths over the 100 base instances for the nearest-neighbour, two-stage, and CA baselines, at vehicle capacity 100 (a) and 150 (b).}
 \label{fig:tourLengths_Cap100_Cap150}
 \end{figure}

 \begin{figure}[H]
 \centering
 \includegraphics[width=0.60\textwidth]{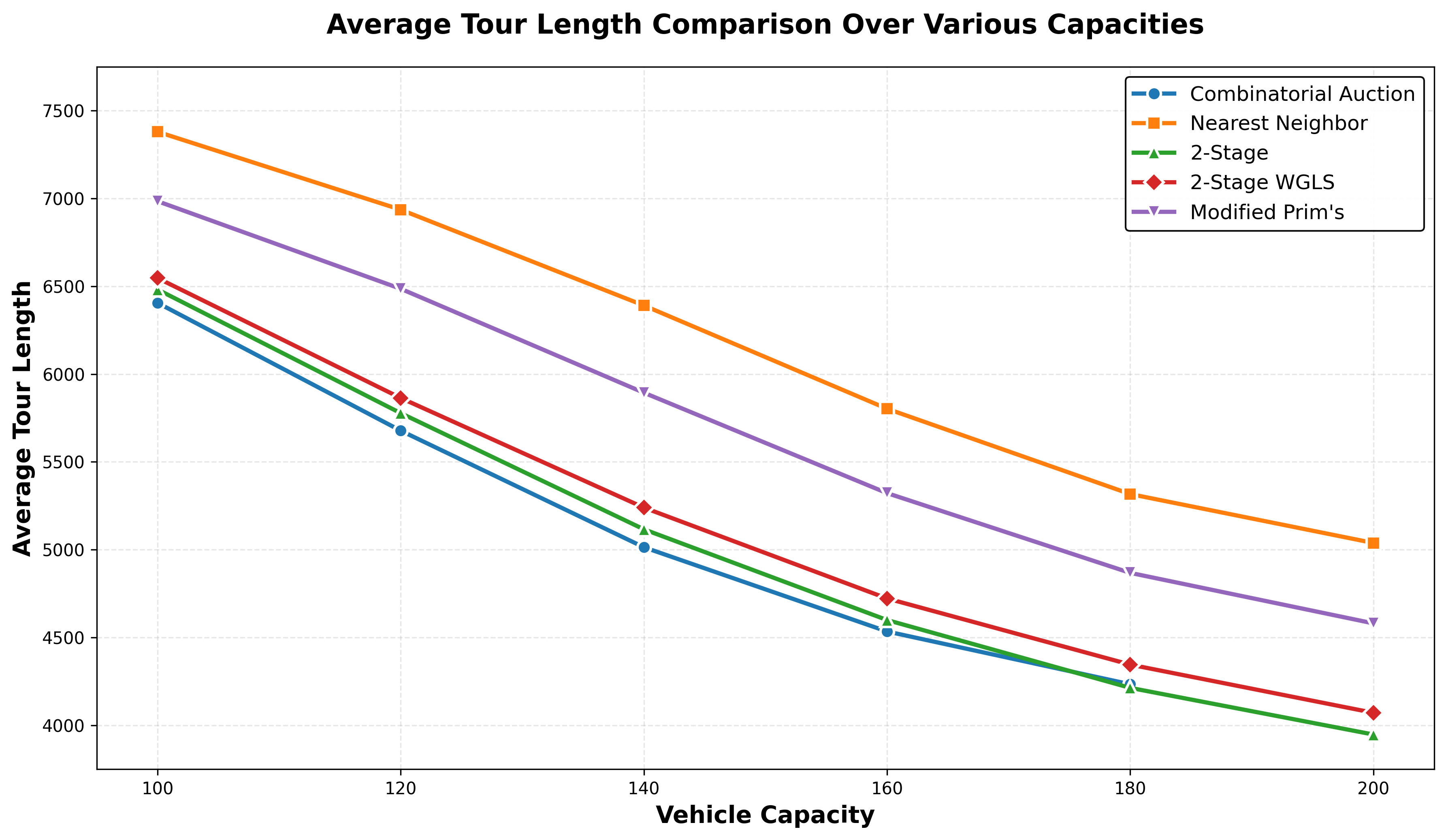}
 \caption{Average tour length as vehicle capacity varies, for the three baselines on the base instances of Figure~\ref{fig:tourLengths_Cap100_Cap150}.}
 \label{fig:capacity_baseline}
 \end{figure}

\noindent one trip whereas we deliberately restrict the auction to small bundles (singletons and pairs). Since that regime is atypical for the capacitated multi-depot setting we study, we do not regard it as a genuine reversal.

The results above clearly establish the Combinatorial Auction as the strongest of the three baselines. We therefore adopt it as our single point of comparison going forward: the larger-scale experiments of \cref{sec:results} pit the graph matching methods, Cluster-First and Match-First, directly against it.

\section{Algorithms and Method Description}
    We now describe in full the two methods that are the focus of this work: a graph matching based method and a Combinatorial Auction based approach. As established in the preliminary analysis (\cref{sec:prelim}), the Combinatorial Auction is the strongest of the standard baseline solvers, so it serves as our point of comparison for the graph matching methods in the larger experiments that follow.
    \subsection{\bf Combinatorial Auction}\label{sec:combinatorialAuction}
    We formulate the MDCVRP as a straightforward combinatorial auction, described below.
    \begin{itemize}
    \item Let $V_1,V_2,\ldots,V_K$ be the set of vehicles presumed to be located at depots. Let $t_1,t_2,\ldots,t_n$ be the targets to be served.
    \item  Let $I_1, I_2, \ldots, I_P$ be the various itemsets, where each itemset is a subset of targets. For example, an itemset $I_j = \{t_1,t_2,t_3\}$ represents the set of targets $t_1,t_2,t_3$.
    \item Each vehicle bids for every itemset, thus the number of bids placed by each vehicle equals the number of itemsets $P$. The bid by each vehicle is the total distance covered in an optimal tour serving the corresponding set of targets. Consider for instance vehicle $V_i$ stationed at depot i. and itemset $I_j = \{t_1,t_2,t_3\}$. The value of the bid made by vehicle $V_i$ for this itemset would be the total distance covered in the optimal tour starting from and ending at depot $i$ and visiting $t_1,t_2,t_3$ exactly once, provided the total weight of all the items in the itemset is less than or equal to the capacity of the vehicle $V_i$ and a high value (eg. $\infty$) otherwise. We collect these bids in a vehicle bid matrix (Table~\ref{bestBidsEachItemsetGeneral}(a)).
   \item To guarantee a feasible solution the itemsets include all singletons: i.e. $\{t_1\},\{t_2\},\ldots,\{t_n\}$
Additionally we considered  all pairs: i.e. $\{t_1,t_2\},\ldots,\{t_{n-1},t_n\}$ and triplets in some of our experiments. Note that in general, calculating the total cost of an optimal tour is the Traveling Salesman Problem (TSP) which is an NP-complete problem. However, since we consider only singletons and other small subsets of targets, the optimal tour cost can be computed in constant time for each itemset by exhaustively considering all possibilities.
\end{itemize}

\begin{table}[h]
  \centering
\begin{tabular}{cc}
  C = \begin{tabular}{c||c|c|c|c|c|}
       &  $I_1$ & $I_2$ & ... & ... & $I_P$ \\ \hline \hline
     $V_1$  & $b_{1,1}$ & $b_{1,2}$ & $\ldots$ & $\ldots$ & $b_{1,P}$ \\ \hline
     ... & ... & ... & ... & ... & ... \\ \hline
    $V_i$  & $b_{i,1}$ & $b_{i,2}$ & $\ldots$ & $\ldots$ & $b_{i,P}$ \\ \hline
    ... & ... & ... & ... & ... & ... \\ \hline
    $V_K$  & $b_{K,1}$ & $b_{K,2}$ & $\ldots$ & $\ldots$ & $b_{K,P}$ \\ \hline
\end{tabular}&
  \begin{tabular}{|c|c|c|c|c|c|} \hline
       Itemset, &  &  &  &  &  \\
       Vehicle & $I_1, V_{i_1}$ & $I_2, V_{i_2}$ & ... & ... & $I_P, V_{i_P}$\\ \hline \hline
       Value of  &  & & & & \\
       lowest bid & $b_1^*$ & $b_2^*$ & $\ldots$ & $\ldots$ & $b_P^*$ \\ \hline
\end{tabular}\\
(a) & (b)\\
\end{tabular}
  \caption{ (a) Cost matrix where $(i,j)^{th}$ entry represents the cost of vehicle $V_i$ servicing the targets in itemset $I_j$ (b) Values of best (lowest) bid  for each of the itemsets (ties broken randomly). Here $V_{i_1}$ is the best (lowest) bidder for itemset $I_1$, the value of the lowest bid for $I_1 = b_1^*$ and so on}
  \label{bestBidsEachItemsetGeneral}
\end{table}

\noindent Computing the optimal route assignments via an Integer Linear Program would involve $K \times P$ variables. Considering only the best (lowest) bids for each set of targets (itemset) reduces the number of variables from $K \times P$ to P, where P the number of itemsets. This effectively collapses the matrix in table \ref{bestBidsEachItemsetGeneral}(a) to that in table \ref{bestBidsEachItemsetGeneral}(b). It can be shown that the Integer Linear Program (ILP) specified in (\ref{ILP_ReducedVariables}) is exactly equivalent to formulating an ILP with the same objective over the set of $K \times P$ variables (proof skipped due to space constraints). We assess two implementations of the combinatorial auction, an Integer Linear Program and a Maximum Independent Set formulation, described next. The two return identical solutions on almost all instances, but the Maximum Independent Set implementation runs significantly faster.

\begin{enumerate}
    \item {\bf Integer Linear Program Formulation:}
Let
\begin{itemize}
\item ${\bf x} = [x_1,\ldots,x_P]^T$ denote the set of (binary) variables to be solved for, where,
\[x_j = \left\{\begin{array}{cl}
          1 & \mbox{if itemset $I_j$ is selected as part of the winning bids}\\
          0 & \mbox{otherwise}
          \end{array}\right.\]
\item ${\bf b} = [b_1^*, b_2^*,\ldots,b_P^*]^T$ where $b_1^*, b_2^*,\ldots,b_P^*$ are the values of the best (lowest) bid  for each of the itemsets (ties broken randomly) as described in table \ref{bestBidsEachItemsetGeneral}(b).

\item Let $M$ be an item membership matrix of size $n \times P$ as specified in (\ref{eqn:itemMembershipMatrix}).
\end{itemize}
\begin{eqnarray}\label{eqn:itemMembershipMatrix}
    M = \begin{tabular}{c||c|c|c|c|c|}
       &  $I_1$ & $I_2$ & ... & ... & $I_P$ \\ \hline \hline
     $t_1$  &  &  & $\ldots$ & $\ldots$ &  \\ \hline
     ... & ... & ... & ... & ... & ... \\ \hline
    $t_n$  &  &  & $\ldots$ & $\ldots$ &  \\ \hline
\end{tabular} &  M(i,j)  =  \left\{ \begin{array}{l}
                        \mbox{1 if target $t_i$ is included in itemset $I_j$}\\
                        \mbox{0 otherwise}
                        \end{array} \right.  &
\end{eqnarray}

The Integer Linear Program that solves for ${\bf x}$ can be written as follows.
\begin{eqnarray}\label{ILP_ReducedVariables}
  \mbox{min } {\bf b}^T {\bf x} & &\\ \nonumber
   \mbox {s.t.}  &  M {\bf x} = {\bf 1} & x_i \in \{0,1\} ~~ i = 1,\ldots,P\\ \nonumber
  \end{eqnarray}
$M {\bf x} = {\bf 1}$ is a system of $n$ equalities, where the $i^{th}$ equation together with the constraint that all $x_j$s are binary ensures that exactly one of the itemsets containing target $t_i$ is selected.

\item {\bf Maximum Independent Set Implementation:}
Because each target must be visited exactly once, the selected itemsets must be pairwise disjoint: no two winning bids may share a target, $I_j \cap I_k = \emptyset, \hspace{1 mm} j \neq k$. The combinatorial auction for the MDCVRP is formulated as a cost-minimization problem, which can be interpreted as selecting a collection of feasible bundles with minimum total cost. In our implementation, we enumerate the feasible bundles, compute the best bid for each one, and then construct a solution greedily by scanning the bids in a fixed order and accepting only those that do not overlap with any bundle already selected. This initial solution is then refined through local search, where selected bundles are revisited one by one and replaced with cheaper compatible alternatives whenever doing so reduces the total cost while preserving feasibility.

\end{enumerate}

\subsection{\bf Matching-based Approach}
We propose a graph-matching approach for the MDCVRP, motivated by the matching-based lower bound previously introduced for the Traveling Tournament Problem (TTP), which schedules a double round-robin tournament in which each team plays every other once at home and once away~\cite{Diptendu2021}. The approach comes in two variants, \emph{Cluster-First} (Algorithm~\ref{alg:cluster_first}) and \emph{Match-First} (Algorithm~\ref{alg:match_first}). Both reduce routing to a minimum-weight matching on the complete graph whose vertices are the targets and whose edge weights are the inter-target distances; we compute it with a Blossom-style routine, LEMON's \textsc{MaxWeightedMatching} on shifted distance weights~\cite{Edmonds1965_Blossom}, which we denote \textsc{MinWeightMatching()}. On a complete graph this returns the minimum-weight (near-)perfect matching, pairing as many targets as possible and leaving at most one unpaired; on a feasibility-restricted graph it may leave more targets unpaired. Each matched pair becomes one tour that leaves a depot, visits both targets, and returns; each unpaired target is served by a singleton round trip; and a tour visiting more than two targets is sequenced by enumerating its few orderings and keeping the cheapest.

\textbf{Cluster-First} partitions the targets before matching: each target is assigned to its nearest depot, and every depot's cluster is solved independently as a single-depot instance, matching within the cluster and routing the matched pairs and leftover singletons from that depot. When a cluster has an odd number of targets the matching leaves one target unpaired, served as a singleton; equivalently, a zero-cost dummy node at the depot keeps the matching perfect and optimal (\cref{lem:dummy}).

\textbf{Match-First} defers the depot assignment: it runs \textsc{MinWeightMatching()} once over all targets, without reference to the depots, and only then routes each group (a matched pair or a leftover singleton) from the depot nearest its centroid. Matching globally lets it form pairings that Cluster-First, confined to within-cluster pairs, cannot.

When the tour size is not restricted to two, the single matching is replaced by the iterative \textsc{MatchAndMerge} procedure (Algorithm~\ref{alg:match_and_merge}). After the first matching, each pair is collapsed to a \emph{pair-node} at its midpoint and a second matching is run on the pair-nodes; two matched pair-nodes are merged into one tour of four targets when their combined demand fits the capacity, and otherwise into the largest feasible three-target tour, with the fourth target returned to the pool, or left as the two original pairs if no three of them fit. The procedure iterates on the still-unmerged targets until no further merge is possible, and any target left over is served as a singleton. Cluster-First runs \textsc{MatchAndMerge} within each cluster; Match-First runs it once over all targets and routes each resulting group from the depot nearest its centroid.

\begin{algorithm}[H]
\small
\caption{\textsc{MatchAndMerge}: Iterative Clustering for Scenario 2.}
\label{alg:match_and_merge}
\begin{algorithmic}[1]
\STATE \textbf{Input:} Target set $\mathcal{U}$, demands $w_i$, distance $d(\cdot,\cdot)$, capacity $C$
\STATE $\mathcal{L}\leftarrow\mathcal{U}$, $\mathcal{G}\leftarrow\emptyset$
\WHILE{$|\mathcal{L}|\ge 2$}
    \STATE Build complete graph $H_1$ on $\mathcal{L}$
    \STATE $M_1\leftarrow\textsc{MinWeightMatching}(H_1)$
    \STATE $\mathcal{L}'\leftarrow$ targets unmatched by $M_1$; replace each $(i,j)\in M_1$ by a midpoint pair-node
    \IF{fewer than two pair-nodes exist}
        \STATE Add all pair-node member sets to $\mathcal{G}$; $\mathcal{L}\leftarrow\mathcal{L}'$; \textbf{continue}
    \ENDIF
    \STATE Build complete graph $H_2$ on pair-node midpoints
    \STATE $M_2\leftarrow\textsc{MinWeightMatching}(H_2)$
    \FOR{each matched pair of pair-nodes $\{(a,b),(c,d)\}\in M_2$}
        \IF{$w_a+w_b+w_c+w_d\le C$}
            \STATE Add $\{a,b,c,d\}$ to $\mathcal{G}$
        \ELSE
            \STATE Let $(G^\star,\ell^\star)$ be the feasible 3-target merge with maximum demand
            \STATE If it exists, add $G^\star$ to $\mathcal{G}$ and $\ell^\star$ to $\mathcal{L}'$; otherwise add $\{a,b\}$ and $\{c,d\}$
        \ENDIF
    \ENDFOR
    \STATE Add pair-node member sets unmatched by $M_2$ to $\mathcal{G}$; $\mathcal{L}\leftarrow\mathcal{L}'$
\ENDWHILE
\STATE Add each remaining target in $\mathcal{L}$ as a singleton in $\mathcal{G}$
\RETURN $\mathcal{G}$
\end{algorithmic}
\end{algorithm}

\begin{algorithm}[H]
\small
\caption{Cluster-First Matching.}
\label{alg:cluster_first}
\begin{algorithmic}[1]
\STATE \textbf{Input:} Depots $\mathcal{D}$, targets $\mathcal{T}$, demands $w_i$, distance $d(\cdot,\cdot)$
\STATE \textbf{Parameters:} Capacity $C$, scenario $\mathcal{S}\in\{\texttt{scenario1},\texttt{scenario2},\texttt{scenario3}\}$
\STATE $\mathcal{R}\leftarrow\emptyset$
\STATE Assign each target $i\in\mathcal{T}$ to $a(i)=\arg\min_{p\in\mathcal{D}}d(p,i)$
\FOR{each depot $p\in\mathcal{D}$}
    \STATE $\mathcal{T}_p\leftarrow\{i\in\mathcal{T}:a(i)=p\}$
    \IF{$\mathcal{S}=\texttt{scenario2}$}
        \STATE $\mathcal{G}_p\leftarrow\textsc{MatchAndMerge}(\mathcal{T}_p,w,d,C)$
        \FOR{each group $G\in\mathcal{G}_p$}
            \STATE Order $G$ by starting at the target nearest to $p$ and minimizing the remaining order
            \STATE Add route $(p\rightarrow G\rightarrow p)$ to $\mathcal{R}$
        \ENDFOR
    \ELSE
        \STATE Build complete graph $H_p$ on $\mathcal{T}_p$ with edge weights $d(i,j)$
        \IF{$\mathcal{S}=\texttt{scenario3}$}
            \STATE Remove edges $(i,j)$ with $w_i+w_j>C$
        \ENDIF
        \IF{$|\mathcal{T}_p|$ is odd}
            \STATE Add a dummy node $\bar\delta$ at $p$ with $d(p,\bar\delta)=0$, $d(k,\bar\delta)=d(p,k)$ \COMMENT{see \cref{lem:dummy}}
        \ENDIF
        \STATE $M_p\leftarrow\textsc{MinWeightMatching}(H_p)$
        \STATE Add route $(p\rightarrow i\rightarrow j\rightarrow p)$ for each real pair $(i,j)\in M_p$
        \STATE Add route $(p\rightarrow i\rightarrow p)$ for each $i$ matched to $\bar\delta$ or left unmatched
    \ENDIF
\ENDFOR
\RETURN $\mathcal{R}$
\end{algorithmic}
\end{algorithm}

\begin{algorithm}[H]
\small
\caption{Match-First Routing.}
\label{alg:match_first}
\begin{algorithmic}[1]
\STATE \textbf{Input:} Depots $\mathcal{D}$, targets $\mathcal{T}$, demands $w_i$, distance $d(\cdot,\cdot)$
\STATE \textbf{Parameters:} Capacity $C$, scenario $\mathcal{S}\in\{\texttt{scenario1},\texttt{scenario2},\texttt{scenario3}\}$
\STATE $\mathcal{R}\leftarrow\emptyset$
\IF{$\mathcal{S}=\texttt{scenario2}$}
    \STATE $\mathcal{G}\leftarrow\textsc{MatchAndMerge}(\mathcal{T},w,d,C)$
\ELSE
    \STATE Build complete graph $H$ on $\mathcal{T}$ with edge weights $d(i,j)$
    \IF{$\mathcal{S}=\texttt{scenario3}$}
        \STATE Remove edges $(i,j)$ with $w_i+w_j>C$
        \STATE Remove edges $(i,j)$ with saving $g_{ij}\le 0$ \COMMENT{see \cref{prop:mf-approx}}
    \ENDIF
    \STATE $M\leftarrow\textsc{MinWeightMatching}(H)$
    \STATE $\mathcal{G}\leftarrow$ matched pairs in $M$ plus unmatched targets as singletons
\ENDIF
\FOR{each group $G\in\mathcal{G}$}
    \STATE Let $\bar{x}_G$ be the centroid of $G$
    \STATE $p_G\leftarrow\arg\min_{p\in\mathcal{D}}d(p,\bar{x}_G)$
    \IF{$\mathcal{S}=\texttt{scenario2}$}
        \STATE Order $G$ by starting at the target nearest to $p_G$ and minimizing the remaining order
    \ENDIF
    \STATE Add route $(p_G\rightarrow G\rightarrow p_G)$ to $\mathcal{R}$
\ENDFOR
\RETURN $\mathcal{R}$
\end{algorithmic}
\end{algorithm}

\noindent We evaluate both variants under three scenarios, each set by a bound on the target demands $w_i$ relative to the common vehicle capacity $C$ and by a bound on the tour size (the number of targets served in a single tour). Depots may dispatch multiple tours; the analysis uses exact metric distances, while the implementation rounds Euclidean distances to two decimals before scaling them to integers.
    \begin{itemize}
        \item \textbf{Scenario 1} ($w_i \leq C/2$, tour size $\leq 2$): every pair is capacity-feasible, so \textsc{MinWeightMatching()} runs on the full complete graph.
        \item \textbf{Scenario 2} ($w_i \leq C/2$, tour size unbounded): tours are built with \textsc{MatchAndMerge}, yielding groups of up to four targets.
        \item \textbf{Scenario 3} ($w_i \leq C$, tour size $\leq 2$): a pair may now exceed the capacity, so every edge $(i,j)$ with $w_i+w_j>C$ is removed before matching. For Match-First we additionally remove each edge for which serving the two targets as singletons from their nearest depots costs no more than routing them together from the depot nearest their centroid.
    \end{itemize}
For a like-for-like comparison, in each scenario the combinatorial auction is restricted to bundles no larger than the largest tour the matching method can form: singletons and pairs in Scenarios~1 and~3, and up to four targets in Scenario~2. As shown in \cref{app:approx}, the tour-size-$\le2$ regimes (Scenarios~1 and~3) are solvable \emph{exactly} in polynomial time, and Cluster-First and Match-First are fast $2$-approximations of that optimum.

\section{Theoretical Analysis}
We summarise the analysis of the matching methods here; full statements and proofs are in Appendices~\ref{app:complexity} and~\ref{app:approx}. Distances are metric, $n$ is the number of targets, and $K$ the number of depots.
\begin{itemize}
    \item \textbf{Running time.} Each routing step reduces to a single Blossom matching. Match-First runs one global matching over all $n$ targets and then assigns each resulting group to a depot, costing $O(n^3+nK)$ in Scenarios~1 and~2. Scenario~3 adds a \emph{best-depot precomputation}, since its savings filter scores each of the $\Theta(n^2)$ candidate pairs against all $K$ depots, raising the cost to $O(n^3+n^2K)$ (\cref{thm:mf-time}). Cluster-First avoids that precomputation: it assigns each target, rather than each pair, to its nearest depot and matches only within clusters, for $O(nK+\sum_p n_p^3)$, where $n_p$ is the size of depot $p$'s cluster (\cref{thm:cf-time}). Because $\sum_p n_p^3\ll n^3$ once targets spread across depots, Cluster-First is the faster of the two, and both are far cheaper than the auction, whose bundle scoring costs $O(Kn^s)$ and grows steeply with the largest bundle size $s$.

    \item \textbf{Exact optimum for tour size two.} The cost of a pairing splits into a fixed round-trip cost minus the distance the pairing saves, so minimising cost is the same as maximising total saving, a matching problem. With a single depot this makes both methods optimal (\cref{thm:single-opt}). More strikingly, for \emph{any} number of depots the tour-size-$\le2$ problem (Scenarios~1 and~3) is solvable exactly in $O(n^2K+n^3)$ time, by scoring each pair from its best depot and taking a maximum-saving matching (\cref{thm:md-exact}). Tour size two is therefore a tractability boundary for the otherwise NP-hard MDCVRP. The combinatorial auction, in its exact ILP form, computes this same optimum, which is why it is the strongest quality baseline (\cref{sec:prelim}).

    \item \textbf{Approximation guarantees.} Let $\Delta$ be the sum of nearest-depot distances; it lower-bounds the cost of any routing (\cref{thm:md-lb}). Cluster-First costs at most $2\Delta$ and is a $2$-approximation, with the factor $2$ tight because a pair whose two targets are nearest to different depots is split across them (\cref{thm:cf-approx}). Match-First matches globally, so without a safeguard it can pair targets that are naturally served from different depots and be arbitrarily far from optimal (\cref{rem:mf-unbounded}); restricting the matching to beneficial, positive-saving pairs, the Scenario~3 filter, restores a $2$-approximation (\cref{prop:mf-approx}). For the four-target tours of Scenario~2 every merge only lowers cost, and Cluster-First is a $4$-approximation, a factor that is tight for nearest-depot clustering (\cref{thm:cf-s2}).

    \item \textbf{Why we use the approximations.} We want a routing that is near-optimal yet extremely fast. The exact method is more expensive even asymptotically, since it pays the $O(n^2K)$ best-depot precomputation and a global $O(n^3)$ matching, an order of magnitude or more of extra work for a quality gain the experiments show is small (\cref{sec:results}), and it must be rebuilt from scratch whenever a target is inserted, whereas Cluster-First re-solves only the affected cluster. Cluster-First and Match-First are therefore a deliberate tradeoff between accuracy and speed (\cref{rem:md-exact}).
\end{itemize}
These guarantees agree with the experiments of \cref{sec:results}: the matching methods lose little in tour quality while running far faster and scaling better than the auction.

\section{Results and Experimental Details}\label{sec:results}
\label{sec:exp-setup}
We run all experiments on a Mac M4 Pro machine with 48\,GB of RAM, using a C++ implementation of the methods above. We generate instances at random: we sample target and depot locations on a 2D Cartesian grid, draw customer demands at random, and take straight-line (Euclidean) distances between locations. We assume a homogeneous fleet, so all vehicles share a common capacity, and we reuse the same list of capacities across datasets to keep results comparable. We use two families of instances: a \emph{base} collection of 100 datasets, each with 100 targets and 10 depots (\cref{sec:prelim,sec:rq3-base,sec:rq4-capacity}), and a larger \emph{scalability} collection of 50 datasets, each with 1000 targets and 20 depots (\cref{sec:rq5-scalability}). Within each dataset we fix the number of depots and targets and vary only the target locations and demands.

\noindent Having selected the Combinatorial Auction as the strongest baseline in \cref{sec:prelim}, the remaining experiments compare it against our graph matching based methods through three questions.
\begin{itemize}
    \item \textbf{RQ1}: How do the graph matching based methods compare to the best baseline in terms of solution quality and computation time?
    \item \textbf{RQ2}: Does the computation-time advantage of the graph matching based methods persist across different vehicle capacities?
    \item \textbf{RQ3}: How well do the graph matching based methods scale to larger problem instances?
\end{itemize}

\subsection{RQ1: Graph Matching versus Combinatorial Auction}\label{sec:rq3-base}
On the base datasets (\cref{sec:exp-setup}) we compare the two matching methods (Cluster-First, Match-First) against the two combinatorial-auction implementations (ILP and Maximum Independent Set) across the three scenarios of \cref{sec:combinatorialAuction}. Figure~\ref{fig:base_boxplots} shows the resulting tour-distance and computation-time distributions. The auction attains slightly shorter tours, but the matching methods run far faster. In Scenario~3, for example, CA-ILP gives the shortest mean tour ($5305$) yet takes $479$\,ms per instance, whereas Match-First and Cluster-First stay within $2$--$6\%$ of it ($5411$ and $5626$) at $4.8$ and $3.3$\,ms, about two orders of magnitude faster. The same pattern holds in Scenarios~1 and~2.

\subsection{RQ2: Effect of Capacity on Computation Time}\label{sec:rq4-capacity}
Figure~\ref{fig:capacity_plots} shows how computation time changes as vehicle capacity is varied for the two matching-based and two CA-based implementations, on the same base datasets. The computation-time advantage of the matching-based methods observed in RQ1 persists over a wide range of vehicle capacities.

\begin{figure}[H]
\centering
\includegraphics[width=0.95\textwidth]{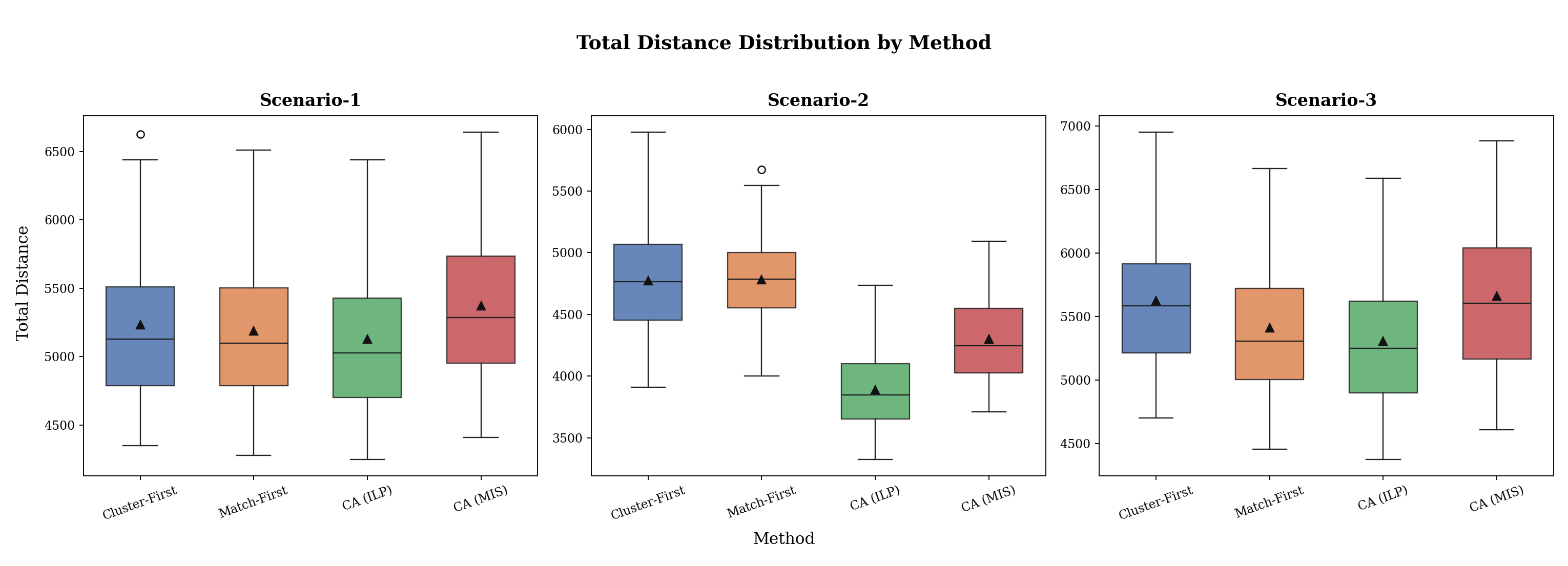}\\[0.4em]
\includegraphics[width=0.95\textwidth]{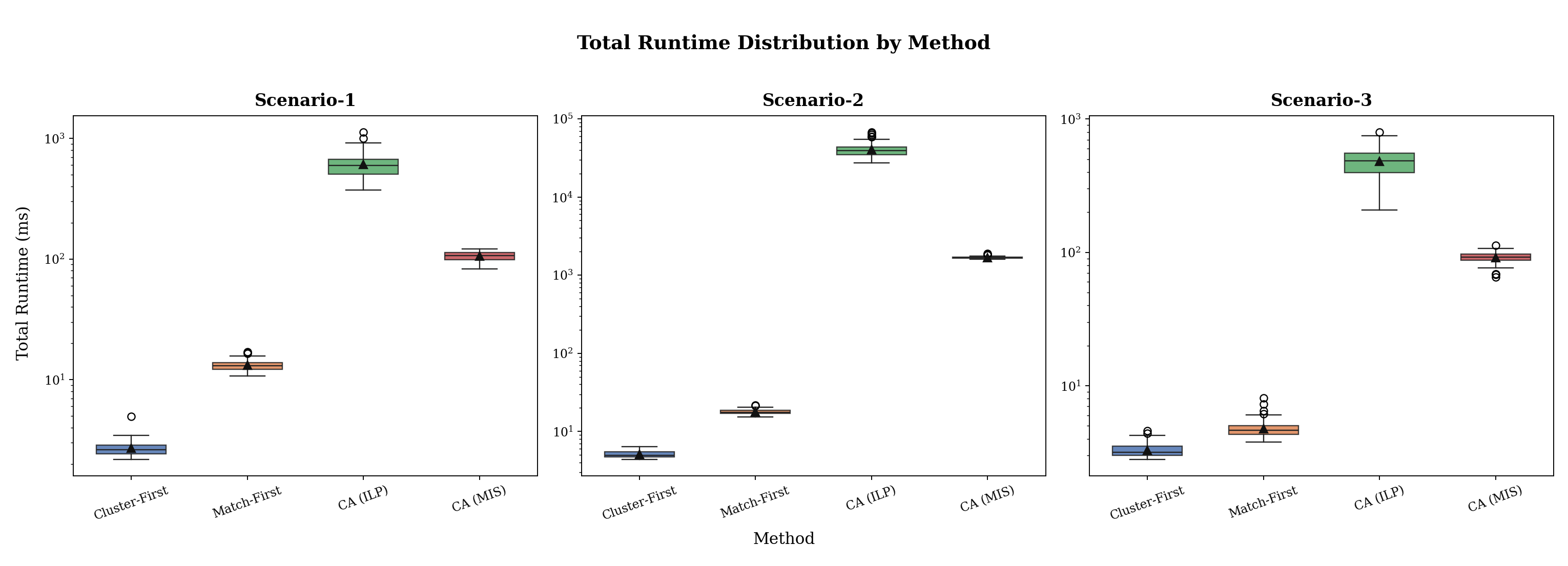}
\caption{Tour distance (top) and computation time (bottom) over the 100 base instances (100 targets, 10 depots), for Cluster-First, Match-First, CA-ILP, and CA-MIS.}
\label{fig:base_boxplots}
\end{figure}

\begin{figure}[H]
\centering
\includegraphics[width=0.95\textwidth]{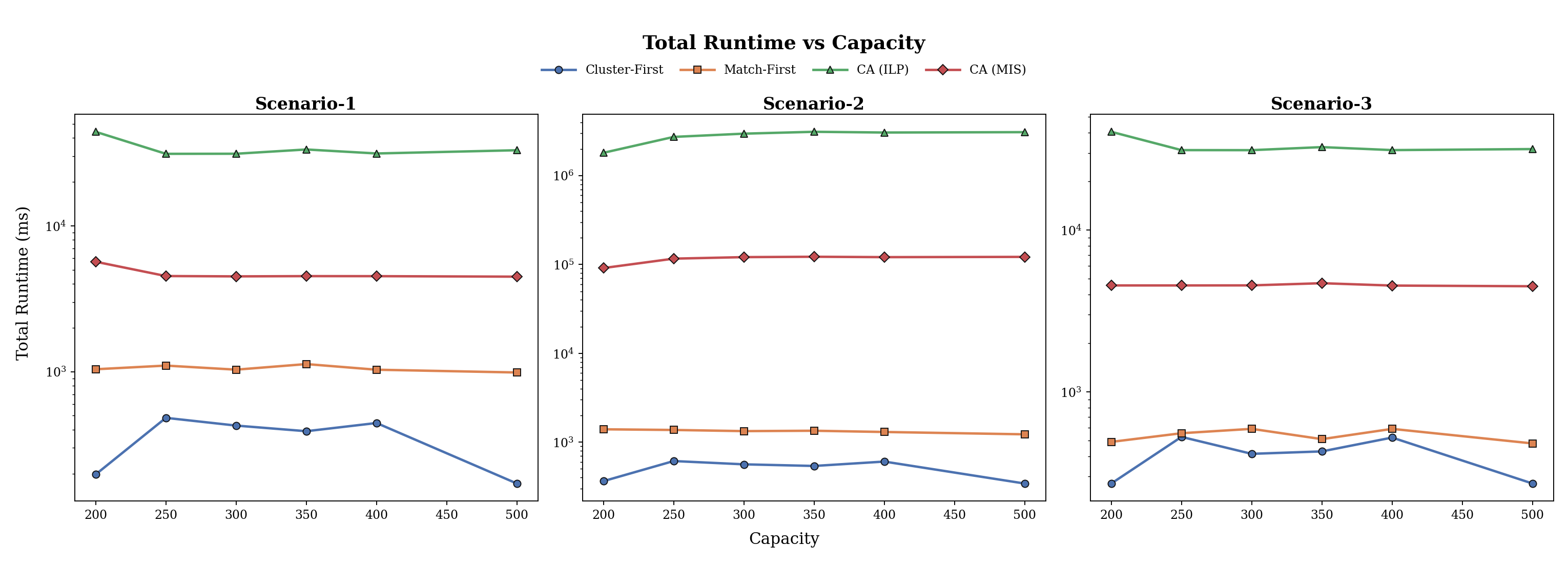}
\caption{Computation time versus vehicle capacity, for Cluster-First, Match-First, CA-ILP, and CA-MIS on the base instances.}
\label{fig:capacity_plots}
\end{figure}

\subsection{RQ3: Scalability}\label{sec:rq5-scalability}
On the scalability datasets (\cref{sec:exp-setup}; 50 datasets, 1000 targets, 20 depots) we again compare only the graph matching based and Combinatorial Auction based approaches. For the scalability plots we report only the MIS-based CA implementation, because the ILP-based implementation does not complete at this problem size. We also omit Scenario~2 from these plots: in this scenario the CA formulation must consider bundles of up to four targets, and with 1000 targets the number of triples alone is $\binom{1000}{3}\approx 1.66\times 10^8$, before the pruned four-target candidates are considered. The resulting scalability distributions are shown in Figure~\ref{fig:scalability_boxplots}.

\noindent At this scale the matching methods are no longer behind on quality. In Scenario~1 the mean tour length is $32{,}617$ for Cluster-First and $32{,}513$ for Match-First, against $33{,}137$ for CA-MIS; in Scenario~3 it is $33{,}215$ and $32{,}979$, against $34{,}360$. The matching methods thus match or slightly beat the auction. The computation-time gap, meanwhile, is large: Cluster-First averages $69$--$83$\,ms per instance and Match-First about $1.1$\,s, whereas CA-MIS takes $31$--$37$\,s, $2$--$3$ orders of magnitude slower, and CA-ILP does not finish at all. Matching therefore scales far better than the auction as the instance grows.

\begin{figure}[H]
\centering
\includegraphics[width=0.90\textwidth]{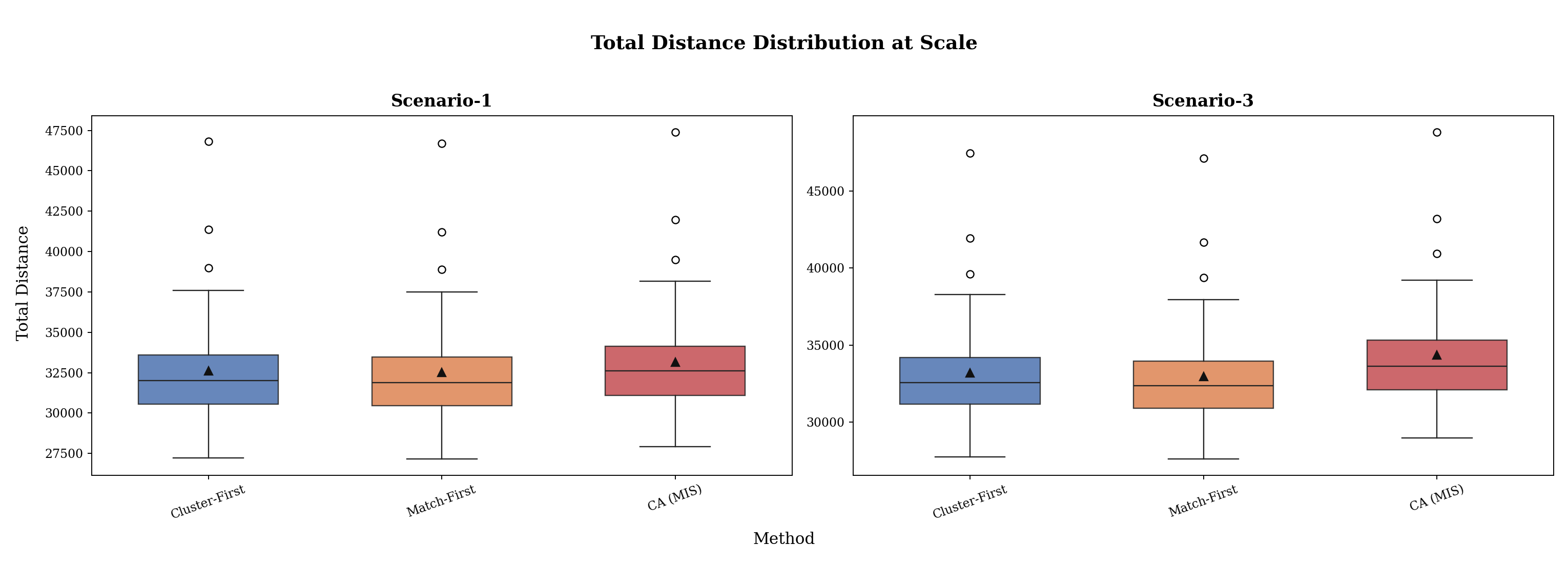}\\[0.4em]
\includegraphics[width=0.90\textwidth]{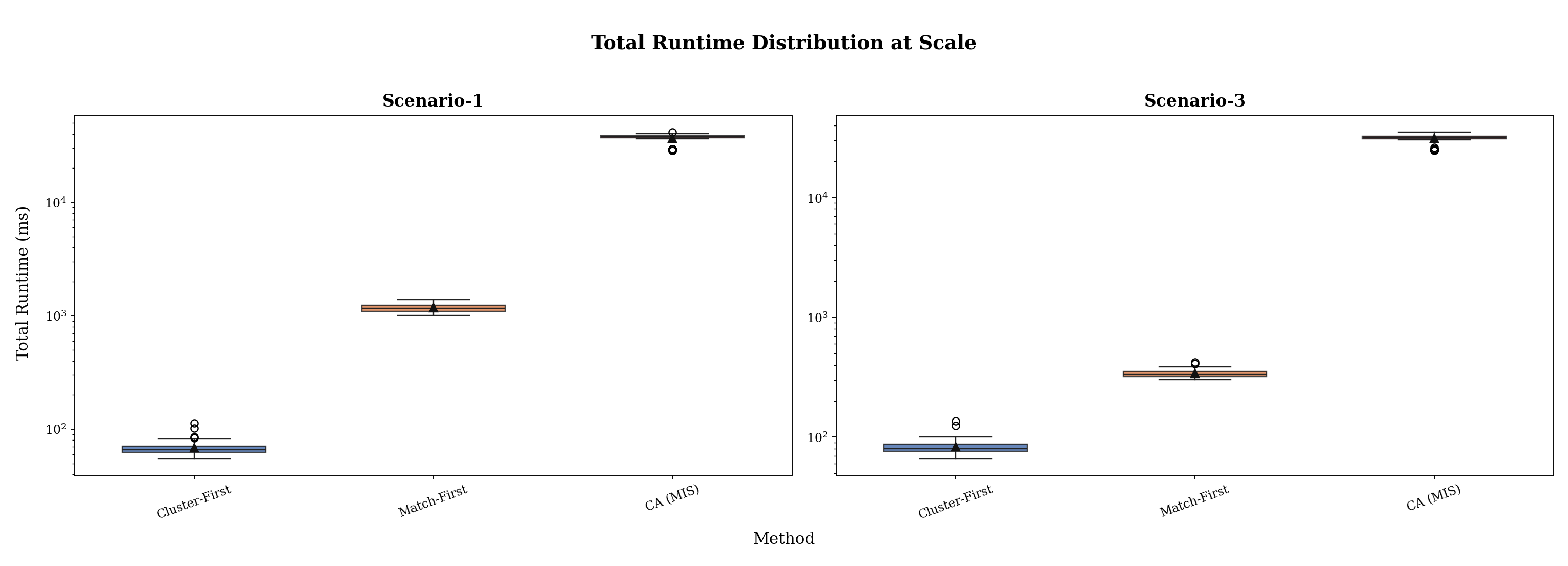}
\caption{Tour distance (top) and computation time (bottom) over the scalability instances (1000 targets, 20 depots), for Cluster-First, Match-First, and CA-MIS in Scenarios~1 and~3. CA-ILP and Scenario~2 are omitted (\cref{sec:rq5-scalability}).}
\label{fig:scalability_boxplots}
\end{figure}

\section{Conclusions and Future Work}
We study the Multi-Depot Capacitated Vehicle Routing Problem through the lens of graph matching and introduce two algorithms, Cluster-First and Match-First, that reduce routing to a sequence of minimum-weight matchings. The central finding is that they are near-optimal yet extremely fast. On the theory side, for tours of up to two targets the matching formulation solves the MDCVRP exactly in polynomial time for any number of depots (\cref{thm:md-exact}); both algorithms are $2$-approximations in the structured regimes, with the factor of two tight (\cref{thm:cf-approx,prop:mf-approx}); and this matching optimum coincides with the exact combinatorial-auction optimum, so the auction we benchmark against is a strong quality baseline. Empirically, the two methods track that baseline in tour length and pull far ahead in speed. On instances of $1000$ customers and $20$ depots they match or slightly beat CA-MIS in tour length, roughly $2\%$ shorter in Scenario~1 and $3$--$4\%$ shorter in Scenario~3, while running two to three orders of magnitude faster, at a scale where the ILP auction does not finish (\cref{sec:rq3-base,sec:rq5-scalability}). Because Cluster-First routes each depot independently, and the nearest-depot assignment avoids the global best-depot precomputation that an exact solve needs, the approach also suits the dynamic setting that motivates it: when a customer arrives, we re-solve only the affected cluster.

\noindent The approach leaves clear room to grow, and each direction is a natural piece of future work. Our exact and $2$-approximation guarantees currently cover tours of up to two targets (Scenarios~1 and~3); for the four-target tours of Scenario~2 we prove single-depot dominance of the tour-size-two optimum and a $4$-approximation (\cref{thm:scenario2,thm:cf-s2}), so extending the guarantees to larger tour sizes, and closing the implementation-versus-theory gap in Scenarios~2 and~3 (\cref{rem:s3-gap}), is a clear next step. Richer models such as heterogeneous fleets and time windows are also within reach. Empirically, we plan to validate on established MDVRP benchmark libraries with repeated trials and confidence intervals, extending the randomly generated Euclidean instances used here. More broadly, because matching-based routing is cheap to run and to update, it is an attractive inner loop for auction- and search-based solvers and a natural building block for fully dynamic vehicle-routing systems.

	\bibliography{references,MDVRPReferences}
    \section*{Appendix}
	\appendix

\section{Detailed Time-Complexity Analysis}\label{app:complexity}

Throughout this appendix we write $n=|\mathcal{T}|$ for the number of targets and $K=|\mathcal{D}|$ for the number of depots. All routing subroutines reduce to computing minimum-weight matchings, for which we use Edmonds' Blossom algorithm~\cite{Edmonds1965_Blossom}.

\begin{lemma}\label{lem:blocks}
On $m$ vertices, $\textsc{MinWeightMatching}$ (Blossom) runs in $O(m^3)$ time. Building the complete distance graph and removing capacity-infeasible edges each cost $O(m^2)$. Ordering a group of bounded size $\le s_0$ costs $O(1)$, by enumerating the $O(s_0!)$ tours over the group and its depot.
\end{lemma}

\begin{proof}
The cubic Blossom bound is standard~\cite{Edmonds1965_Blossom,Gabow1976}. A complete graph on $m$ points has $O(m^2)$ edges, each weighted by one distance evaluation, so construction and the infeasibility scan are $O(m^2)$. Every group in Scenarios~1--3 has size at most $s_0=4$, so a constant number of orderings suffices.
\end{proof}

\begin{theorem}[Match-First complexity]\label{thm:mf-time}
Algorithm~\ref{alg:match_first} (Match-First) runs in $O(n^3+nK)$ time in Scenarios~1 and~2, and in $O(n^3+n^2K)$ time in Scenario~3.
\end{theorem}

\begin{proof}
\emph{Scenario 1.} The algorithm builds one complete graph $H$ on all $n$ targets ($O(n^2)$ by \cref{lem:blocks}) and runs a single $\textsc{MinWeightMatching}$ on $H$ ($O(n^3)$). It then forms at most $n/2$ pairs and at most $n$ singleton groups; for each it computes a centroid ($O(1)$) and scans the $K$ depots for the nearest ($O(K)$), totalling $O(nK)$. Summing gives $O(n^3+nK)$.

\emph{Scenario 3.} We also delete edges $(i,j)$ with $w_i+w_j>C$ ($O(n^2)$) and run the savings filter, a \emph{best-depot precomputation} that scores each of the $\Theta(n^2)$ pairs against the $K$ depots at $O(n^2K)$. With the matching ($O(n^3)$) and depot assignment ($O(nK)$), the total is $O(n^3+n^2K)$. Scenarios~1--2 avoid the $O(n^2K)$ term by assigning depots per group, not per pair.

\emph{Scenario 2.} The work is dominated by $\textsc{MatchAndMerge}$ (Algorithm~\ref{alg:match_and_merge}). Each iteration of its while-loop at least quarters the number $\ell=|\mathcal{L}|$ of lone targets: the level-1 matching yields $\lfloor \ell/2\rfloor$ pair-nodes, and the level-2 matching pairs these into $\lfloor \ell/4\rfloor$ groups, each returning at most one leftover target (when a $3$-target group forms). Hence $|\mathcal{L}'|\le 1+\lfloor \ell/4\rfloor$. One iteration runs two matchings on $\le\ell$ vertices, costing $O(\ell^3)$, so the total is
\[
\sum_{k\ge 0} O\!\left(\Big(\tfrac{n}{4^{k}}\Big)^{3}\right)=O(n^3)\sum_{k\ge0}\frac{1}{64^{k}}=O(n^3),
\]
and there are $O(\log n)$ iterations. Depot assignment and bounded-size ordering of the resulting $O(n)$ groups add $O(nK)$ by \cref{lem:blocks}. The total is again $O(n^3+nK)$.
\end{proof}

\begin{theorem}[Cluster-First complexity]\label{thm:cf-time}
Let $n_p=|\mathcal{T}_p|$ be the number of targets assigned to depot $p$, so $\sum_{p\in\mathcal{D}}n_p=n$. Algorithm~\ref{alg:cluster_first} (Cluster-First) runs in
\[
O\!\left(nK+\textstyle\sum_{p\in\mathcal{D}}n_p^{3}\right)
\]
time. In the worst case (all targets in one cluster) this is $O(n^3+nK)$, and for balanced clusters with $n_p=\Theta(n/K)$ it is $O(nK+n^3/K^2)$.
\end{theorem}

\begin{proof}
Assigning each target to its nearest depot scans all $K$ depots, costing $O(nK)$. This assignment is \emph{per target}, so Cluster-First never scores the $\Theta(n^2)$ pairs against the depots and incurs no $O(n^2K)$ term; within a cluster the depot is fixed to $p$. Each depot $p$ is then an independent single-depot subinstance on $\mathcal{T}_p$, costing $O(n_p^3)$ by \cref{lem:blocks} (Scenarios~1 and~3) or by the geometric sum of \cref{thm:mf-time} (Scenario~2). Summing gives $\sum_p O(n_p^3)$; convexity of $x\mapsto x^3$ gives the worst case $\sum_p n_p^3\le n^3$ and the balanced case $\sum_p (n/K)^3=n^3/K^2$.
\end{proof}

\begin{remark}\label{rem:cf-vs-mf}
\Cref{thm:mf-time,thm:cf-time} explain the observed running-time gap. Match-First runs one global cubic matching ($\Theta(n^3)$) and, in Scenario~3, an $O(n^2K)$ best-depot scan; Cluster-First does neither, decomposing into per-depot matchings ($\sum_p n_p^3\ll n^3$) and assigning each target to its nearest depot in $O(nK)$. Match-First stays useful because its global matching can form pairings Cluster-First cannot (\cref{rem:mf-unbounded}). An auction enumerating bundles of size $\le s$ instead pays $O(Kn^s)$, which grows far faster in $s$.
\end{remark}

\section{Optimality and Approximation Analysis}\label{app:approx}

We now analyse solution quality. Distances form a metric satisfying the triangle inequality, the only property used below (the centroid rule of \cref{thm:mf-multi} also uses Euclidean midpoints). For tour size $\le2$ (Scenarios~1 and~3) we solve the problem exactly in polynomial time, for a single depot (\cref{thm:single-opt}) and for any number of depots (\cref{thm:md-exact}), and Cluster-First and Match-First are $2$-approximations of that optimum (\cref{thm:cf-approx,prop:mf-approx}). For Scenario~2 exactness is lost, but Cluster-First remains a $4$-approximation (\cref{thm:cf-s2}).

\paragraph*{Notation}
For a single depot $p$ let $D_p=\sum_{k\in\mathcal{T}}d(p,k)$ be the total depot-to-target distance. For a matching $M$ write $W(M)=\sum_{(i,j)\in M}d(i,j)$, and for even $n$ let $W_m=\min_M W(M)$ over perfect matchings. In the multi-depot case let $\delta_k=\min_{q\in\mathcal{D}}d(q,k)$ be the nearest-depot distance of target $k$, $a(k)=\arg\min_q d(q,k)$ its nearest depot, and $\Delta=\sum_{k\in\mathcal{T}}\delta_k$. We write $\mathrm{OPT}_s$ for the optimal cost over routings whose tours serve at most $s$ targets, and abbreviate $\mathrm{OPT}_2$.

\subsection{From TTP-2 to routing}

The analysis adapts the matching-based lower bound for {\sc TTP-2}~\cite{Diptendu2021}. Let $W_i$ be the total distance from team $i$ to the others, $W_t$ the total edge weight, and $W_m$ the weight of a minimum-weight matching. A team may visit its two matched opponents on one away tour, so its minimum travel is $W_i+W_m$; summing over teams gives the lower bound
\begin{equation}\label{eq:ttp-lb}
\sum_{i\in V}(W_i+W_m)=2W_t+nW_m .
\end{equation}
For {\sc TTP-2} this is only a lower bound, since the teams' schedules cannot all be realised at once. Routing has no such constraint, so the matching bound is \emph{achieved}. The mechanism is a \emph{savings} accounting: the cost of any pairing is a fixed round-trip cost minus the distance saved by combining targets into pairs.

\begin{lemma}[Savings identity]\label{lem:savings}
Fix a single depot $p$ and a tour-size-$\le2$ routing whose served pairs form a matching $M$ of $\mathcal{T}$, every unmatched target being served by its own round trip. Writing the \emph{saving} of a pair as $s_{ij}:=d(p,i)+d(p,j)-d(i,j)\ge0$,
\[
\mathrm{cost}(M)=2D_p-\sum_{(i,j)\in M}s_{ij}.
\]
If $M$ is a perfect matching this equals $D_p+W(M)$.
\end{lemma}

\begin{proof}
Serving each target by its own round trip costs $\sum_k 2d(p,k)=2D_p$. Replacing the two round trips of $i$ and $j$ by the joint tour $p\to i\to j\to p$ changes the cost by
\[
\big[d(p,i)+d(i,j)+d(j,p)\big]-\big[2d(p,i)+2d(p,j)\big]=-s_{ij},
\]
and $s_{ij}\ge0$ by the triangle inequality; summing over $M$ gives the identity. If $M$ is perfect then $\sum_{(i,j)\in M}\big(d(p,i)+d(p,j)\big)=D_p$, so substituting the definition of $s_{ij}$ gives $\mathrm{cost}(M)=D_p+W(M)$.
\end{proof}

\subsection{Single-depot analysis}

\subsubsection{Scenarios 1 and 3 (tour size at most two)}

\begin{theorem}[Single-depot optimality; approximation ratio $1$]\label{thm:single-opt}
For a single depot and tour size $\le2$, the routing that serves the pairs of a \emph{maximum-saving} matching (a maximum-weight matching under the weights $s_{ij}$, over the capacity-feasible pairs) and every other target as a singleton is optimal. In Scenario~1 (all pairs feasible, since $w_i\le C/2$) with even $n$, this optimum equals
\[
\mathrm{OPT}_2=D_p+W_m,
\]
attained by the minimum-weight perfect matching that \textsc{MinWeightMatching} computes. In Scenario~3 ($w_i\le C$) the same maximum-saving matching, taken over the graph with the infeasible edges ($w_i+w_j>C$) removed, is optimal.
\end{theorem}

\begin{proof}
By \cref{lem:savings}, $\mathrm{cost}(M)=2D_p-\Sigma(M)$ with $\Sigma(M)=\sum_{(i,j)\in M}s_{ij}$ and $2D_p$ fixed, so minimising cost is exactly maximising the total saving $\Sigma(M)$. As the savings are non-negative, this is a maximum-weight matching problem, which the Blossom routine solves exactly; the returned routing is therefore optimal. In Scenario~1 with even $n$ the graph is complete, so a maximum-saving matching may be taken perfect; there $\Sigma(M)=D_p-W(M)$ (\cref{lem:savings}), so maximising the saving is minimising $W(M)$, whose minimum is $W_m$, giving cost $2D_p-(D_p-W_m)=D_p+W_m$. Scenario~3 removes only the infeasible edges before matching, leaving Steps~1--2 unchanged, so the maximum-saving matching on the restricted graph is again optimal.
\end{proof}

\begin{lemma}[Odd $n$: the depot dummy]\label{lem:dummy}
On a perfect matching $\sum_{(i,j)\in M}\big(d(p,i)+d(p,j)\big)=D_p$ is constant, so minimising $W(M)$ and maximising $\Sigma(M)$ coincide. When $n$ is odd, adjoin a dummy target $\bar\delta$ at the depot, with $d(p,\bar\delta)=0$ and $d(k,\bar\delta)=d(p,k)$. Its partner $t^\star$ is then served by the round trip $p\to t^\star\to\bar\delta\to p$, a singleton at zero saving, and the minimum-weight perfect matching of the augmented graph jointly selects the target to leave alone and the optimal pairing of the rest.
\end{lemma}

\begin{proof}
Perfectness gives $\sum_{(i,j)\in M}(d(p,i)+d(p,j))=D_p$, hence $\Sigma(M)=D_p-W(M)$ and $\max_M\Sigma(M)=\min_M W(M)$. The dummy edge $(t^\star,\bar\delta)$ has weight $d(t^\star,\bar\delta)=d(p,t^\star)$ and saving $d(p,t^\star)+d(p,\bar\delta)-d(t^\star,\bar\delta)=0$, so it encodes a zero-saving singleton; the augmented matching weight equals $D_p-\Sigma(M)$, which the matcher minimises.
\end{proof}

\begin{remark}[The implementation minimises distance, not saving]\label{rem:s3-gap}
The shipped \textsc{MinWeightMatching} minimises pair distance (LEMON's maximum-weight matcher on shifted weights). On a complete graph with even $n$ this coincides with saving maximisation (\cref{lem:dummy}), so the implementation attains the optimum of \cref{thm:single-opt} in Scenario~1. On the feasibility-restricted graph of Scenario~3 the two objectives can differ, since a short feasible edge may carry a smaller saving than a longer one, so attaining single-depot optimality there requires the matcher in maximum-saving mode. The $2$-approximation guarantees below use only $s_{ij}\ge0$ and are unaffected.
\end{remark}

\subsubsection{Scenario 2 (tours of up to four targets)}

\begin{lemma}[A merge never increases cost]\label{lem:merge}
Let two matched pairs $A$ and $B$ with total demand at most $C$ be served from a single depot $p$. Relabel their four targets so that $a_1$ is a depot-nearest target, $a_2$ its partner in $A$, and $B=\{b_1,b_2\}$. The merged tour $p\to a_1\to a_2\to b_1\to b_2\to p$ is shorter than the two separate pair tours by exactly the super-edge saving
\[
\sigma(A,B)=\big(d(a_2,p)+d(p,b_1)\big)-d(a_2,b_1)\ \ge\ 0 .
\]
The nearest-first ordering computed by Algorithms~\ref{alg:cluster_first} and~\ref{alg:match_first} does at least as well, since it starts at the depot-nearest target and hence has this tour in its search space.
\end{lemma}

\begin{proof}
The two separate tours total $d(p,a_1)+d(a_1,a_2)+d(a_2,p)+d(p,b_1)+d(b_1,b_2)+d(b_2,p)$, while the merged tour is $d(p,a_1)+d(a_1,a_2)+d(a_2,b_1)+d(b_1,b_2)+d(b_2,p)$. Their difference is $d(a_2,b_1)-d(a_2,p)-d(p,b_1)=-\sigma(A,B)$, and $\sigma(A,B)\ge0$ by the triangle inequality. As the algorithms fix the depot-nearest first target and minimise over the remaining orders, the tour they return is no longer than this admissible one.
\end{proof}

\begin{theorem}[Scenario 2, single depot]\label{thm:scenario2}
For a single depot with $n$ even, suppose every merge performed by $\textsc{MatchAndMerge}$ is a capacity-feasible four-target merge (no $3$-target split occurs). Then the Scenario-2 routing has cost
\[
\mathrm{cost}_{\mathrm{S2}}\ \le\ \mathrm{OPT}_2-\sum_{\text{merges}}\sigma(A,B)\ \le\ \mathrm{OPT}_2=D_p+W_m,
\]
i.e.\ it is never worse than the optimal tour-size-$\le2$ routing, and strictly better whenever some merge has $\sigma(A,B)>0$.
\end{theorem}

\begin{proof}
The level-1 matching of $\textsc{MatchAndMerge}$ is a minimum-weight perfect matching, so by \cref{thm:single-opt} the pre-merge routing costs $\mathrm{OPT}_2=D_p+W_m$. Each four-target merge replaces two pair tours by one and lowers the total by at least $\sigma(A,B)\ge0$ (\cref{lem:merge}), while unmerged pairs keep their cost. Telescoping over the merges gives the bound.
\end{proof}

\begin{remark}
When four demands exceed $C$, Algorithm~\ref{alg:match_and_merge} forms a three-target group and re-queues one target; unlike a four-merge, this split is not guaranteed cost-non-increasing. Scenario~2 is thus a heuristic refinement that provably dominates the tour-size-$\le2$ optimum whenever every merge is feasible (e.g.\ $w_i\le C/4$) and performs strongly otherwise (\cref{sec:results}).
\end{remark}

\subsection{Multi-depot analysis}

\subsubsection{Scenarios 1 and 3 (tour size at most two)}

\begin{theorem}[Multi-depot lower bound]\label{thm:md-lb}
For any number of depots and tour size $\le2$, every feasible routing has cost at least $\Delta=\sum_{k\in\mathcal{T}}\delta_k$.
\end{theorem}

\begin{proof}
The tours partition $\mathcal{T}$. A pair tour $\{i,j\}$ from a depot $q$ costs $d(q,i)+d(i,j)+d(j,q)\ge d(q,i)+d(q,j)\ge\delta_i+\delta_j$, and a singleton for $k$ costs $2d(q,k)\ge\delta_k$. Each target thus contributes at least $\delta_k$ to the tour serving it, and summing over the partition gives cost $\ge\Delta$.
\end{proof}

This bound is in fact achieved: with tours capped at two, the routing decomposes into a matching with independently chosen depots, so the problem is polynomially solvable, as we now show.

\begin{theorem}[Exact multi-depot optimum in polynomial time; approximation ratio $1$]\label{thm:md-exact}
For any number of depots and tour size $\le2$ (Scenarios~1 and~3), define for every capacity-feasible pair $\{i,j\}$ the \emph{best-depot pair cost} and \emph{true saving}
\[
c_{ij}=\min_{q\in\mathcal{D}}\big[d(q,i)+d(i,j)+d(q,j)\big],
\qquad
g^{*}_{ij}=2\delta_i+2\delta_j-c_{ij},
\]
and let $G^{*}=\max_M\sum_{(i,j)\in M}g^{*}_{ij}$ over matchings of the feasible-pair graph. Then
\[
\mathrm{OPT}_2=2\Delta-G^{*},
\]
attained by serving the pairs of a maximum-$g^{*}$ matching, each from a depot achieving $c_{ij}$, and every unmatched target as a singleton from its nearest depot. It is computable in $O(n^2K+n^3)$ time.
\end{theorem}

\begin{proof}
Any feasible routing is a matching $M$ of served pairs together with singletons for the rest, and the depot of each tour may be chosen independently of the others. A pair $\{i,j\}$ costs at least $c_{ij}$ (equality at a minimising depot); a singleton $k$ costs at least $2\delta_k$ (equality at its nearest depot). Hence
\[
\mathrm{OPT}_2=\min_M\Big[\sum_{(i,j)\in M}c_{ij}+\!\!\sum_{k\ \text{unmatched}}\!\!2\delta_k\Big]
=\min_M\Big[2\Delta-\sum_{(i,j)\in M}g^{*}_{ij}\Big]=2\Delta-G^{*},
\]
using $c_{ij}=2\delta_i+2\delta_j-g^{*}_{ij}$ to re-account the matched targets' singleton terms. A maximum-weight matching under $g^{*}$ attains $G^{*}$ (negative-weight edges never enter it), and the stated routing realises that cost.

The running time splits into two parts. The \emph{best-depot precomputation} fixes, for every candidate pair, its cheapest serving depot: computing all $\delta_k$ costs $O(nK)$, and scoring each of the $\Theta(n^2)$ pairs against the $K$ depots to obtain the $c_{ij}$ (equivalently the $g^{*}_{ij}$) costs $O(n^2K)$. The \emph{global maximum-weight matching} over all $n$ targets then costs $O(n^3)$ (\cref{lem:blocks}). The total is $O(n^2K+n^3)$.
\end{proof}

\begin{remark}[Why we approximate rather than solve exactly]\label{rem:md-exact}
\Cref{thm:md-exact} places tour size two at a tractability boundary: bounding tours to two targets reduces the NP-hard MDCVRP~\cite{Lenstra1981} to matching. It subsumes \cref{thm:single-opt} ($K=1$ gives $g^{*}_{ij}=s_{ij}$) and sharpens \cref{thm:md-lb} to $\Delta\le\mathrm{OPT}_2\le2\Delta$. We do not ship it, because we want a routing that is near-optimal \emph{and} fast, and the exact method is asymptotically slower for a negligible quality gain. It pays an $O(n^2K)$ best-depot precomputation and a global $O(n^3)$ matching, whereas Cluster-First fixes each target to its nearest depot (the restriction behind its factor-$2$ guarantee) and matches within clusters, in $O(nK+\sum_p n_p^3)$ (\cref{thm:cf-time}). At $n=1000$, $K=20$ and $10^8$ operations per second, the precomputation alone ($\approx2\times10^7$) already exceeds Cluster-First's entire cost ($\approx2.5\times10^6$), and the $O(n^3)\approx10^9$ matching pushes the exact solve to about $10$ seconds per instance against $25$ milliseconds for Cluster-First, a $400\times$ gap. The gap only widens in the dynamic setting, where the exact method rebuilds everything on each insertion while Cluster-First re-solves one cluster. Cluster-First and Match-First are therefore a deliberate choice, not a concession to intractability.
\end{remark}

\begin{theorem}[Cluster-First is a $2$-approximation]\label{thm:cf-approx}
In Scenarios~1 and~3, for any number of depots, Algorithm~\ref{alg:cluster_first} (Cluster-First) returns a routing of cost at most $2\Delta\le2\,\mathrm{OPT}_2$. Its approximation ratio is exactly $2$.
\end{theorem}

\begin{proof}
Cluster-First routes each cluster $\mathcal{T}_p$ independently. Since $p$ is the nearest depot of every $k\in\mathcal{T}_p$, $D'_p:=\sum_{k\in\mathcal{T}_p}d(p,k)=\sum_{k\in\mathcal{T}_p}\delta_k$, so $\sum_p D'_p=\Delta$. By \cref{lem:savings} cluster $p$ costs $2D'_p-\sum s_{ij}\le2D'_p$ (savings are non-negative), for any matching returned. Summing gives $\mathrm{cost}_{\mathrm{CF}}\le2\Delta$, and $\mathrm{OPT}_2\ge\Delta$ (\cref{thm:md-lb}) gives $\mathrm{cost}_{\mathrm{CF}}\le2\,\mathrm{OPT}_2$.

\emph{Tightness.} Place depots at $(0,0)$ and $(2L,0)$ and targets at $(L-\varepsilon,0)$ and $(L+\varepsilon,0)$. Cluster-First assigns the two targets to different depots and serves two singletons of total cost $4(L-\varepsilon)$, whereas pairing them from either depot costs $2L+2\varepsilon$; the ratio tends to $2$ as $\varepsilon/L\to0$. The factor $2$ therefore cannot be improved; the loss comes exactly from pairs that straddle depot regions.
\end{proof}

\begin{remark}[Why no separate savings computation is needed]\label{rem:cf-impl}
On complete cluster graphs (Scenario~1), distance-minimising \textsc{MinWeightMatching} plus the depot dummy of \cref{lem:dummy} already maximises the per-cluster saving, so Cluster-First routes each cluster optimally with no explicit savings computation. In Scenario~3 the two objectives may differ (\cref{rem:s3-gap}), but the $2$-approximation still holds, as its proof uses only $s_{ij}\ge0$.
\end{remark}

\begin{proposition}[Match-First with the savings restriction is a $2$-approximation]\label{prop:mf-approx}
In Scenarios~1 and~3, restrict the global matching to \emph{beneficial} pairs, those with positive centroid saving
\[
g_{ij}=\big(2\delta_i+2\delta_j\big)-\big(d(p_c,i)+d(i,j)+d(p_c,j)\big)>0,
\qquad
p_c=\arg\min_{p\in\mathcal{D}}d(p,\bar x_{ij}),
\]
where $\bar x_{ij}$ is the centroid of $\{i,j\}$, and serve every unmatched target as a singleton. Then Match-First returns a routing of cost at most $2\Delta\le2\,\mathrm{OPT}_2$.
\end{proposition}

\begin{proof}
Serving all targets as singletons from their nearest depots costs $\sum_k 2\delta_k=2\Delta$. Routing a matched pair $\{i,j\}$ from $p_c$ instead changes the total by exactly $-g_{ij}$, and the restriction keeps only pairs with $g_{ij}>0$, so $\mathrm{cost}_{\mathrm{MF}}=2\Delta-\sum_{(i,j)\in M}g_{ij}\le2\Delta$; with $\mathrm{OPT}_2\ge\Delta$ (\cref{thm:md-lb}) this is $\le2\,\mathrm{OPT}_2$. The argument is identical in Scenarios~1 and~3; the implementation applies the filter in Scenario~3 (Algorithm~\ref{alg:match_first}), so the Scenario-1 variant lacks this guarantee unless the filter is enabled there too.
\end{proof}

\begin{proposition}[Unrestricted Match-First cost bounds]\label{thm:mf-multi}
In Scenarios~1 and~3 with even $n$, let $M$ be the global minimum-weight perfect matching, of weight $W_m$. Routing each pair from the depot minimising $d(p,i)+d(p,j)$ gives cost at most $\Delta+2W_m$; the centroid rule of Algorithm~\ref{alg:match_first} gives cost at most $\Delta+3W_m$ on Euclidean instances.
\end{proposition}

\begin{proof}
Fix a pair $(i,j)\in M$.

\emph{Best-depot rule.} Using depot $a(i)$, the tour costs $\delta_i+d(i,j)+d(a(i),j)$, and $d(a(i),j)\le\delta_i+d(i,j)$, so it is at most $2\delta_i+2d(i,j)$; symmetrically at most $2\delta_j+2d(i,j)$, hence at most $(\delta_i+\delta_j)+2d(i,j)$. Summing over $M$ gives $\Delta+2W_m$.

\emph{Centroid rule.} For the midpoint $\bar x$ of $i,j$, $d(p_G,\bar x)\le d(a(i),\bar x)\le\delta_i+\tfrac12 d(i,j)$, so $d(p_G,i),d(p_G,j)\le\delta_i+d(i,j)$ and the tour costs at most $2\delta_i+3d(i,j)$; symmetrically at most $(\delta_i+\delta_j)+3d(i,j)$, giving $\Delta+3W_m$.
\end{proof}

\begin{remark}[Why the restriction is needed]\label{rem:mf-unbounded}
Without the filter, the matching can pair two targets the optimum keeps as singletons. With depots at $(0,0),(L,0)$ and targets at $(1,0),(L-1,0)$, the only perfect matching has weight $W_m=L-2$ and serves its pair from one depot at cost $\Theta(L)$, while the optimum uses two unit singletons at cost $4$: the ratio is unbounded. Here $g_{ij}<0$, so \cref{prop:mf-approx} discards this edge. The same global freedom, though, lets Match-First form pairings Cluster-First forbids: a pair sharing a nearest depot is routed from it at the optimal cost $\delta_i+d(i,j)+\delta_j$ (\cref{thm:single-opt}), which is why Match-First often beats Cluster-First when pairs respect depot regions.
\end{remark}

\subsubsection{Scenario 2 (tours of up to four targets)}

The lower bound $\Delta$ does not extend to four-target tours: $\mathrm{OPT}_4$ can be as small as $\Delta/2$ (four co-located targets at distance $R$ from their nearest depot are served by one tour of length $2R$, while $\sum_k\delta_k=4R$), and exact optimisation over size-$\le4$ tours is no longer a matching problem. A multi-depot guarantee nonetheless survives for Cluster-First.

\begin{theorem}[Cluster-First is a $4$-approximation in Scenario 2]\label{thm:cf-s2}
Let $\mathrm{OPT}_4$ be the optimal cost over routings with tour size $\le4$. In Scenario~2, Algorithm~\ref{alg:cluster_first} (Cluster-First) returns a routing of cost at most $2\Delta\le4\,\mathrm{OPT}_4$; its approximation ratio is $4$. More generally, for tour size $\le s$ it is an $s$-approximation.
\end{theorem}

\begin{proof}
\emph{Upper bound.} For a group $G$ in cluster $\mathcal{T}_p$, the tour $p\to v_1\to\cdots\to v_g\to p$ has length at most $\sum_{k\in G}2d(p,k)=\sum_{k\in G}2\delta_k$: each leg satisfies $d(v_t,v_{t+1})\le d(v_t,p)+d(p,v_{t+1})$ by the triangle inequality, and $p$ is the nearest depot of every $k\in\mathcal{T}_p$. Summing over groups gives $\mathrm{cost}_{\mathrm{CF}}\le2\Delta$.

\emph{Lower bound.} A tour serving a set $S$ with $|S|\le s$ from a depot $q$ has length at least $2\max_{k\in S}d(q,k)\ge\tfrac{2}{|S|}\sum_{k\in S}\delta_k\ge\tfrac{2}{s}\sum_{k\in S}\delta_k$; summing over an optimal solution gives $\mathrm{OPT}_s\ge2\Delta/s$. With $s=4$, $\mathrm{cost}_{\mathrm{CF}}\le2\Delta\le4\,\mathrm{OPT}_4$.
\end{proof}

\begin{remark}[The factor $4$ is essentially tight for nearest-depot clustering]\label{rem:cf-s2-tight}
One might expect \cref{thm:cf-s2} to be loose, since its ingredients $\mathrm{cost}_{\mathrm{CF}}\le2\Delta$ and $\mathrm{OPT}_4\ge\Delta/2$ look tight on opposite instances. It is not. Take four targets at the corners of a square of side $\varepsilon$ around a point $X$, and depots $D_1,\dots,D_4$ with $D_i$ on the ray from $X$ through corner $t_i$, at distance $\delta$ beyond it. Then $t_i$'s unique nearest depot is $D_i$ with $\delta_i=\delta$, so $\Delta=4\delta$. Cluster-First serves each $t_i$ from $D_i$ by a round trip of length $2\delta$, giving $\mathrm{cost}_{\mathrm{CF}}=8\delta=2\Delta$, whereas the optimum serves all four on one tour of length $\to2\delta=\Delta/2$ as $\varepsilon\to0$; the ratio tends to $4$. The looseness is intrinsic to nearest-depot clustering, which cannot serve a mutually close group on one tour when its members prefer different depots. Improving the constant requires relaxing that assignment, which we leave to future work.
\end{remark}

\end{document}